\documentclass{article} 
\usepackage{times}  
\usepackage{url}  
\usepackage{graphicx}  
\usepackage{fullpage}
\usepackage{amsmath,amsthm,amssymb}
\usepackage{caption}
\usepackage{subcaption}
\usepackage{color}

\newcommand{\R}{\mathbb{R}}

\newcommand{\X}{\mathcal{X}}

\renewcommand{\log}{\lg}

\newcommand{\M}{\mathcal{M}}
\newcommand{\tmin}{\textrm{min}}
\newcommand{\COST}{\mathrm{C}}
\newcommand{\OPT}{\mathrm{OPT}}
\DeclareMathOperator*{\argmin}{arg\,min}

\newcommand{\todo}[1]{\textcolor{red}{#1}}
\renewcommand{\todo}[1]{}
\newcommand{\concave}{concave }

\newtheorem{lemma}{Lemma}
\newtheorem{definition}{Definition} 
\newtheorem{Theorem}{Theorem}

\begin{document}

\title{Fast Exact k-Means, k-Medians and Bregman Divergence Clustering in 1D}

\author{
  Allan Gr\o nlund\thanks{Aarhus University.
    Email: \texttt{jallan@cs.au.dk}. Supported by MADALGO - Center for
    Massive Data Algorithmics, a Center of the Danish National
    Research Foundation.
  }    
  \and
  Kasper Green Larsen\thanks{Aarhus University.
    Email: \texttt{larsen@cs.au.dk}. Supported by MADALGO, a Villum
    Young Investigator Grant and an AUFF Starting Grant.
  }
  \and
  Alexander Mathiasen\thanks{Aarhus University.
    Email: \texttt{alexander.mathiasen@gmail.com}. Supported by
    MADALGO and an AUFF Starting Grant.
  }
  \and
  Jesper Sindahl Nielsen
  \thanks{Aarhus University.
    Email: \texttt{jasn@cs.au.dk}. Supported by MADALGO.
  }
  \and
  Stefan Schneider
  \thanks{University of California, San Diego. Email:
    \texttt{stschnei@cs.ucsd.edu}.  Supported by NSF grant CCF-1213151 from the
Division of Computing and Communication Foundations. Any opinions,
findings and conclusions or recommendations expressed in this material are those 
of the authors and do not necessarily reflect the views of the National Science Foundation.}
  \and
  Mingzhou Song
  \thanks{New Mexico State University. Email: \texttt{joemsong@cs.nmsu.edu}
}
}

\date{}
\maketitle
\begin{abstract} 
  The $k$-Means clustering problem on $n$ points is NP-Hard for any dimension $d\ge 2$, however, for the 1D case there exists exact polynomial time algorithms.
  Previous literature reported an $O(kn^2)$ time dynamic programming
  algorithm that uses $O(kn)$ space. It turns out that the problem has
  been considered under a different name more than twenty years ago.
  We present all the existing work that had been overlooked and
  compare the various solutions theoretically. Moreover,
  we show how to reduce the space usage for some of them, as well as
  generalize them to data structures that can quickly report an optimal $k$-Means clustering for any $k$.
  Finally we also generalize all the algorithms to work for the absolute distance and to work for any Bregman Divergence.
  We complement our theoretical contributions by experiments that
  compare the practical performance of the various algorithms.
\end{abstract}

\section{Introduction}
\label{sec:introduction}
Clustering is the problem of grouping elements into clusters such that each element is similar to the elements in the cluster assigned to it
and not similar to elements in any other cluster.  It is one of, if not the, primary problem in the area of machine learning known as
Unsupervised Learning and no clustering problem is as famous and widely considered as the $k$-Means problem:
Given a multiset $\X=\{x_1,...,x_n\}\subset\R^d$ find $k$ centroids $\M=\{\mu_1,...,\mu_k\}\subset \R^d$ minimizing $\sum_{x\in\X}\min_{\mu\in \M}||x-\mu||^2$.
Several NP-Hardness results exist for finding the optimal $k$-Means clustering in general, forcing one to turn towards heuristics.  
$k$-Means is NP-hard even for $k=2$ and general dimension
\cite{Aloise2009} and it is also NP-hard for $d=2$ and general $k$
\cite{Mahajan2009}.  Even hardness of approximation results exist
\cite{Lee201740,DBLP:conf/compgeom/AwasthiCKS15}.
In \cite{DBLP:conf/compgeom/AwasthiCKS15} the authors show there exists an $\varepsilon > 0$ such that it is NP-hard to approximate $k$-Means to within a factor $1+\varepsilon$ of optimal, and in \cite{Lee201740} they prove that $\varepsilon \ge 0.0013$.
On the upper bound side the best known polynomial time approximation algorithm for $k$-Means has an approximation factor of 6.357 \cite{DBLP:journals/corr/AhmadianNSW16}.
In practice, Lloyd's algorithm is a popular iterative local search heuristic that starts from some random or arbitrary clustering. The running time of Lloyd's algorithm is $O(tknd)$ where $t$ is the number of rounds of the local search procedure.
In theory, if Lloyd's algorithm is run to convergence to a local minimum, $t$ could be exponential and there is no guarantee on how well the solution found approximates the optimal solution \cite{Arthur:2006:SKM:1137856.1137880,Vattani2011}.
Lloyd's algorithm is often combined with the effective seeding technique for selecting initial centroids due to
\cite{arthur2007k} that gives an expected $O(\log k)$ approximation ratio for the initial clustering, which can then be further improved by Lloyd's algorithm.

For the one-dimensional case, the $k$-Means problem is not NP-hard. In particular, there is an $O(kn^2)$ time and $O(kn)$ space dynamic programming solution for the 1D case, due to work by \cite{wang2011ckmeans} giving focus to this simpler special case of $k$-Means.
The 1D $k$-Means problem is encountered surprisingly often in practice, some examples being in data analysis in social networks, bioinformatics and retail market \cite{arnaboldi2012analysis,jeske2013genome,pennacchioli2014retail}.
Note that an optimal clustering in 1D simply covers the input points with non-overlapping intervals.


It is natural to try other reasonable distance measures for the data considered and define different clustering problems instead of using the sum of squares of the Euclidian distance that defines $k$-Means.
For instance, one could use any $L_p$ norm instead. The special case of $p=1$ is known as $k$-Medians clustering and has also received considerable attention.
The $k$-Medians problems is also NP-hard in dimensions $2$ and up, and the best polynomial time approximation algorithm has an approximation factor of 2.633 \cite{DBLP:conf/compgeom/AwasthiCKS15}.

In \cite{Banerjee:jmlr} the authors consider and define clustering with Bregman Divergences.
Bregman Divergence generalizes squared Euclidian distance and thus Bregman Clusterings include the $k$-Means problem, as well as a wide range of other clustering problems that can be defined from Bregman Divergences like e.g. clustering with Kullback-Leibler divergence as the cost.
Interestingly, the heuristic local search algorithm for Bregman Clustering~\cite{Banerjee:jmlr} is basically the same approach as Lloyd's algorithm for $k$-Means.
Clustering with Bregman Divergences is clearly NP-Hard as well since it includes $k$-Means clustering.
We refer the reader to \cite{Banerjee:jmlr} for more about the general problem.
For the 1D version of the problem, \cite{Nielsen2014OptimalIC} generalized the algorithm from \cite{wang2011ckmeans} to the $k$-Median
problem and Bregman Divergences achieving the same $O(k n^2)$ time and $O(kn)$ space bounds.

\subsection{Even Earlier Work}
Unknown to both \cite{wang2011ckmeans, Nielsen2014OptimalIC}, it turns out that the $k$-Means problem has been studied before under a different name.
The paper~\cite{Wu_quant} from 1980 considers  a 1D data discrete quantization problem, where the problem is to quantize $n$ (sorted) weighted input points to $k$ representatives, and the goal is to compute the $k$ representatives that minimizes the quantization error.
The error measure for this quantization is weighted least squares, which is a simple generalization of the 1D $k$-Means cost.
Formally, given $n$ points $x_1,\dots, x_n$ with weights $w_1,\dots,w_n$, find centroids $\M=\{\mu_1,...,\mu_k\}\subset \R$ minimizing the cost
\begin{equation}
  \sum_{i=1}^n w_i \min_{\mu\in\M}(x_i-\mu)^2
  \label{eq:km_weighed}
\end{equation}

The algorithm given in \cite{Wu_quant} use $O(kn)$ time (and space). This was later improved to $O(k \sqrt{n \log n})$ time and space in \cite{Aggarwal1994} that also introduces a simple $O(n \lg U)$ time and linear space algorithm where $U$ is a \emph{universe} size that depends on the input.
Finally, the running time was reduced to $n 2^{O(\sqrt{ \lg \lg n \lg k})}$ by an algorithm of Schieber \cite{Schieber} in 1998 for the case $k=\Omega(\lg n)$. This algorithm uses linear space.
This is, to the best of our knowledge, currently the best bound known for 1D $k$-Means. 


\subsection{Our Contribution}

In this paper we present an overview of different 1D $k$-Means algorithms.
We show how to reduce the space for the presented dynamic programming algorithm and how this leads to a simple data structure that uses $O(n)$ space and can report an optimal $k$-Means clustering for any $k$ in $O(n)$ time.
We generalize the fast algorithms for $k$-Means to all Bregman Divergences and the Absolute Distance cost function.
Finally, we give a practical comparison of the algorithms and discuss the outcome.
The main purpose of this paper is clarifying the current state of algorithms for 1D $k$-Means both in theory and practice, introducing a regularized version of 1D $k$-Means, and generalizing the results to more distance measures.
All algorithms and proofs for $k$-Means generalize to handle the weighted version of the $k$-Means cost (Equation \ref{eq:km_weighed}) considered in  \cite{Wu_quant}.
For simplicity, we present the algorithms only for the standard unweighted definition of the $k$-Means cost.


We start by giving the dynamic programming formulation \cite{wang2011ckmeans, Wu_quant} that yields an $O(kn^2)$ time algorithm.
We then describe how \cite{Wu_quant} improves this to an algorithm  that runs in $O(n\log n + k n)$, or $O(k n)$ time if the input is already sorted.
Both algorithms compute the cost of the optimal clustering using $k'$ clusters for all $k' \leq k$. This is relevant for instance for model selection of the right $k$.
Second, we present the reduction from 1D $k$-Means to the shortest path graph problem considered in \cite{Schieber}, that yields a 1D $k$-Means clustering algorithm
that uses $n2^{O(\sqrt{\lg \lg n \lg k})}$ time for $k=\Omega(\lg n)$ and $O(n)$ space. 
In contrast to the $O(kn)$ time algorithm, this algorithm does not compute the optimal costs for using $k'$ clusters for all $k' \leq k$.
Finally, we present an algorithm using $O(n \lg U)$ time and linear space, where $\lg U$ is the number bits used to represent each input point.
%

For comparison, in 1D, Lloyd's algorithm takes $O(t k \lg n)$ time if the input is sorted, where $t$ is the number of rounds before convergence.
For reasonable values of $k$ the time is usually sublinear and extremely fast.
However Lloyd's algorithm does not necessarily compute the optimal clustering, it is only a heuristic. The algorithms considered in this paper compute the \emph{optimal} clustering quickly, so even for large $n$ and $k$ we do not need to settle for approximation or uncertainty about the quality of the clustering.

The $n2^{O(\sqrt{\lg \lg n \lg k})}$ and the $O(n\lg U)$ time algorithms for 1D $k$-Means are based on a natural regularized
version of $k$-Means clustering where instead of specifying the number of clusters beforehand, we instead specify a per cluster cost and then minimize the cost of the clustering plus the cost of the number of clusters used.
Formally, the problem is as follows: Given $\X=\{x_1,...,x_n\}\subset \R$ and a non-negative real number $\lambda \in \mathbb{R}^+$, compute the optimal regularized clustering:
\begin{equation*}
  \argmin_{k, \M=\{\mu_1,...,\mu_k\}} \sum_{x\in\X}\min_{\mu\in\M}(x-\mu)^2 + \lambda k
\end{equation*}
Somewhat surprisingly, it takes only $O(n)$ time to find the solution to the regularized $k$-Means if the input is sorted.



The $k$-Medians problem is to compute a clustering that minimizes the sum of absolute distances to the centroid, i.e. compute  $\M=\{\mu_1,...,\mu_k\}\subset \R$ that minimizes
\begin{align*}
	\sum_{x\in\X}\min_{\mu\in\M}\lvert x-\mu \rvert
\end{align*}
We show that the $k$-Means algorithms generalize naturally to solve this problem in the same time bounds as for the $k$-means problem.

Let $f$ be a differentiable real-valued strictly convex function.
The Bregman Divergence $D_f$ induced by $f$ is
\begin{align*}
  D_f(x,y) = f(x) - f(y) - \nabla_f(y)(x-y)
\end{align*}
Notice that the Bregman Divergence induced from $f(x) = x^2$, gives squared Euclidian Distance (k-Means).
Bregman divergences are not metrics since they are not symmetric in general and the triangle inequality is not necessarily satisfied.
They do have other redeeming qualities, for instance Bregman Divergences are convex in the first argument, albeit not the second, see \cite{Banerjee:jmlr,Boissonnat2010} for a more comprenhensive treatment.

The Bregman Clustering problem as defined in
\cite{Banerjee:jmlr} is to find $k$ centroids $\M = \{\mu_1,...,\mu_k\}$ that minimize
\begin{align*}
  \label{equ:bregman_clustering}
  \sum_{x \in \mathcal{X}} \min_{\mu \in \M} D_f(x,\mu)
\end{align*}
where $D_f$ is a Bregman Divergence.
For our case, where the inputs $x,y \in \mathbb{R}$, we assume that
computing a Bregman Divergence, i.e. evaluating $f$ and its derivative,
takes constant time.
We show that the $k$-Means algorithms naturally generalize to 1D clustering using any Bregman Divergence to define the cluster cost while still
maintaing the same running time as for $k$-Means.

\paragraph{Implementation.}
An independent implementation of the $O(n \log n+kn)$ time algorithm
is available in the R package Ckmeans.1d.dp~\cite{ckmeans}. This is the algorithm first presented in \cite{Wu_quant}. The
implementation is for $k$-Means clustering, and uses $O(kn)$ space.


\section{Algorithms for 1D $k$-Means - A Review} 
\label{sec:old_alg}
In the following we assume sorted input $x_1\le...\le x_n\in \R$. 
Notice that there could be many ways of partitioning the input and computing centroids that achieve the same cost.
This is for instance the case if the input is $n$ identical points.
The task at hand is to find any optimal solution for 1D $k$-Means.

Let $CC(i,j)=\sum_{\ell=i}^j(x_\ell-\mu_{i,j})^2$ be the cost of grouping
$x_i,...,x_j$ into one cluster with the optimal choice of centroid,
$\mu_{i,j} =\frac{1}{j-i+1}\sum_{\ell=i}^jx_\ell$, the mean
of the points.
\begin{lemma}
  It takes $O(n)$ time and space to construct a data structure that computes $CC(i,j)$ in constant time
  \label{lemma:precompute}
\end{lemma}
\begin{proof}
This is a standard application of prefix sums.
By definition,
\begin{align*}
  CC(i,j) & = \sum_{\ell=i}^{j} (x_\ell - \mu_{i,j})^2
  = \sum_{\ell=i}^{j} x_\ell^2 + \mu_{i,j}^2 - 2x_\ell \mu_{i,j}
   =  (j-i+1)\mu_{i,j}^2 + \mu_{i,j} \sum_{\ell=i}^{j} x_\ell + \sum_{\ell=i}^{j} x_\ell^2 .
\end{align*}
Using prefix sum arrays of $x_1,\dots,x_n$ and $x_1^2,\dots,x_n^2$ computing the centroid $\mu_{i,j}$ and 
the sums takes $O(1)$ time.
\end{proof}
\subsection{The $O(k n^2)$ Dynamic Programming Algorithm}
\label{subsec:sketch}
The algorithm finds the optimal clustering using $i$ clusters for all prefixes of the points $x_1,\dots,x_m$, for $m=1,\dots,n$, and $i=1,\dots, k$ with Dynamic Programming as follows.
Let $D[i][m]$ be the cost of optimally clustering $x_1,...,x_m$ into $i$ clusters. If $i=1$ the cost of optimally clustering $x_1,...,x_m$ is the cluster cost $CC(1,m)$. That is, $D[1][m]=CC(1,m)$ for all $m$.
By Lemma~\ref{lemma:precompute}, this takes $O(n)$ time.

For $i>1$
\begin{equation}
	D[i][m]=\min_{j=1}^m D[i-1][j-1] + CC(j, m)
	\label{equ:rec_cost}
\end{equation}
Notice that $D[i-1][j-1]$ is the cost of optimally clustering $x_1,...,x_{j-1}$ into $i-1$ clusters and $CC(j,m)$ is the cost of clustering $x_j,...,x_m$ into one cluster.
This makes $x_j$ the first point in the last and rightmost cluster.
Let $T[i][m]$ be the argument that minimizes \eqref{equ:rec_cost}
\begin{equation}
	T[i][m]:=\arg\min_{j=1}^m D[i-1][j-1]+ CC(j,m)
	\label{equ:rec_argmin}
\end{equation}
It is possible there exists multiple $j$ obtaining same minimal value for \eqref{equ:rec_cost}. To make the optimal clustering unique, such ties are broken in favour of smaller $j$.

Notice $x_{T[i][m]}$ is the first point in the rightmost cluster of the optimal clustering. Thus, given $T$ one can find the optimal solution by standard backtracking.
One can naively compute each entry of $D$ and $T$ using \eqref{equ:rec_cost} and \eqref{equ:rec_argmin}. This takes $O(n)$ time for each cell, thus $D$ and $T$ can be computed in $O(kn^2)$ time using $O(kn)$ space.
This is exactly what is described in \cite{wang2011ckmeans}. This dynamic programming recursion is also presented in \cite{Wu_quant}.



\subsection{An $O(kn)$ time algorithm}
\label{sec:fully_monotone}
The Dynamic Programming algorithm can be sped up significantly to $O(kn)$ time by reducing the time to compute each row of $D$ and $T$ to $O(n)$ time instead of $O(n^2)$ time.
This is exactly what \cite{Wu_quant} does, in particular it is shown how to reduce the problem of computing a row of $D$ and $T$ to searching an implicitly defined $n\times n$ matrix of a special form, which then allows computing each row of $D$ and $T$ in linear time.

Define $\COST_i[m][j]$ as the cost of the optimal clustering of $x_1,\dots,x_m$ using $i$ clusters, restricted to having the rightmost cluster (largest cluster center) contain the elements $x_j,\dots, x_m$.
For convenience, we define $\COST_i[m][j]$ for $j > m$ as the cost of clustering $x_1,\dots,x_m$ into $i-1$ clusters, i.e. the last cluster is empty.
This means that $\COST_i$ satisfies: 
\begin{align*}
  \COST_i[m][j] & = D[i-1][\min\{j-1,m\}] + CC(j,m)
\end{align*}
where by definition $CC(j,m)=0$ when $j > m$ (which is consistent with the definition in Section~\ref{sec:old_alg}). This menas that $D[i][m]$ relates to $\COST_i$ as follows:
\begin{align*}
  D[i][m] & =\min_j \COST_i[m][j]
\end{align*}
where ties are broken in favor of smaller $j$ (as defined in Section~\ref{subsec:sketch}).

This means that to compute a row of $D$ and $T$, we are computing $\min_j \COST_i[m][j]$ for all $m=1,\dots,n$.
Think of $\COST_i$ as an $n\times n$ matrix with rows indexed by $m$ and columns indexed by $j$.
With this interpretation, computing the $i$'th row of $D$ and $T$  corresponds to computing for each row $r$ in $\COST_i$, the column index $c$ that corresponds to the smallest value in row $r$.
In particular, the entries $D[i][m]$ and $T[i][m]$ correpond to the value and the index of the minimum entry in the $m$'th row of $\COST_i$ respectively.
The problem of finding the minimum value in every row of a matrix has been studied before \cite{Aggarwal1987}.
First we need the definition of a monotone matrix.
\begin{definition}
  \cite{Aggarwal1987}
  Let $A$ be a matrix with real entries and let $\argmin(i)$ be the index of the leftmost column
  containing the minimum value in row $i$ of $A$. $A$ is said to be \emph{monotone} if $a < b$ implies that $\argmin(a) \leq \argmin(b)$.
  $A$ is \emph{totally monotone} if all of its submatrices are monotone.
\end{definition}
In \cite{Aggarwal1987}, the authors showed:
\begin{Theorem}
  \cite{Aggarwal1987} Finding $\argmin(i)$ for each row $i$ of an arbitrary $n \times m$ monotone matrix requires $\Theta(m \log
  n)$ time, whereas if the matrix is totally monotone, the time is
  $O(m)$ when $m> n$ and is $O(m(1+\log(n/m)))$ when $m < n$.  
\end{Theorem}
The fast algorithm for totally monotone matrices is known as the \emph{SMAWK} algorithm.

Let us relate this monotonicity to the 1D $k$-Means clustering problem.
That $\COST_i$ is monotone means that if we consider the optimal clustering of the points $x_1,\dots,x_a$ with $i$ clusters, and we add more points $x_{a+1} \leq \dots \leq x_{b}$ after $x_a$,
then the first (smallest) point in the last (rightmost) of the $i$ clusters can only increase (move right) in the new optimal clustering of $x_1,\dots,x_{b}$.
This sounds like it should be true for 1D $k$-Means and in fact it is.
Thus, applying the algorithm for monotone matrices, one can fill a row of $D$ and $T$ in $O(n\log n)$ time leading to an $O(k n\log n)$ time algorithm for 1D $k$-Means, which is already a great improvement.

However, as shown in \cite{Wu_quant} the matrix $\COST_i$ defined above is in fact totally monotone.
This follows from the following fact about the 1D $k$-Means clustering cost  \todo{redo indicies here}
\begin{Theorem}[Monge Concave for 1D k-Means cost \cite{Wu_quant}]

  For any $a < b$ and $u < v$, it holds that:
  $$
  CC(v, b) + CC(u, a)  \leq  CC(u,b) + CC(v, a).
  $$
\label{thm:wu_monge}
\end{Theorem}
This is the property known as the \textit{concave monge} (concave for short) property \cite{Yao_F,Hirschberg_Larmore,Wilber}.
We prove this for the general case of Bregman Divergences in Section \ref{sec:extensions}.
For completeness we include a proof that Total Monotonicity is implied by concave cost funtion below.
\begin{lemma}
  \label{lem:monotone}
  The matrix $\COST_i$ is totally monotone if the cluster cost $CC(i, j)$ is monge concave
\end{lemma}
\begin{proof}
  As~\cite{Aggarwal1987} remarks, a matrix $A$ is totally monotone if all its $2 \times 2$ submatrices are monotone.
  To prove that $\COST_i$ is totally monotone, we need to prove that for any two row indices $a,b$ with $a < b$ and two column indices $u,v$ with $u < v$, it holds that if $\COST_i[a][v] < \COST_i[a][u]$ then $\COST_i[b][v] < \COST_i[b][u]$.

  Notice that these values correspond to the costs of clustering elements $x_1,\dots, x_{a}$ and $x_1,\dots,x_{b}$, starting the rightmost cluster with element $x_v$ and $x_u$ respectively.
  Since $\COST_i[m][j] = D[i-1][\min\{j-1,m\}] + CC(j,m)$, this is the same as proving that
    \begin{align*}
      & D[i-1][\min\{v-1,m\}] + CC(v,a) <  D[i-1][\min\{u-1,m\}] + CC(u,a)  \Rightarrow\\
      & D[i-1][\min\{v-1,m\}] + CC(v,b) <  D[i-1][\min\{u-1,m\}] + CC(u,b)
    \end{align*}
    which is true if we can prove that $CC(v, b) - CC(v, a) \leq  CC(u,b) - CC(u,a)$. 
    Rearranging terms, what we need to prove is that for any $a < b$ and $u < v$, it holds that:
    %
    $$
    CC(v, b) + CC(u, a)  \leq  CC(u,b) + CC(v, a).
    $$
    %
    which is the monge concave property.
\end{proof}

As explained in \cite{Wu_quant} the total monotonicity property of the cost matrices and the SMAWK algorithm directly yields an $O(kn)$ time (and space) algorithm for 1D $k$-means.
\begin{Theorem}
  \cite{Wu_quant}
  Computing an optimal $k$-Means clustering of a sorted input of size $n$ for takes $O(k n)$ time.
\end{Theorem}
By construction the cost of the optimal clustering is computed for all $k'\leq k$.
By storing the $T$ table with cluster centers then for any $k'\leq k$ the cluster indices of an optimal clustering can be extracted in $O(k')$ time.

The concave propery has been used to significantly speed up algorithms, in particular (Dynamic Programming) algorithms, for several other 1D problems. For the interested reader we refer to \cite{Yao_F,Hirschberg_Larmore,Wilber}. 

\subsection{A $n 2^{O(\sqrt{\lg \lg n \lg k})}$ time algorithm for $k=\Omega(\lg n)$}
\label{sec:fastest_algo}
The concave property of the $k$-Means cost yields and algorithm for computing the optimal $k$-Means clustering for a given $k=\Omega(\lg n)$ in $n 2^{O(\sqrt{\lg \lg n \lg k})}$ time.
The result follows almost directly from Schieber \cite{Schieber} who gives an algorithm with the aforementioned running time for finding the shortest path of fixed length $k$ in a directed acyclic graph with $n$ nodes and weights, $w(i,j)$, that satisfy the concave property. It is assumed that the weights are represented by a function (oracle) that returns the weight of a requested edge in constant time.
\begin{Theorem}[\cite{Schieber}]
  Computing a minimum weight path of length $k$ between any two nodes in a directed acyclic graph of size $n$ where the weights satisfy the concave property takes  $n 2^{O(\sqrt{\lg \lg n \lg k})}$ time using $O(n)$ space.
\end{Theorem}

The 1D $k$-Means problem is reducible to this directed graph problem as follows. Sort the input in $O(n\log n)$ time and let $x_1 \leq x_2 \leq \dots x_n$ denote the sorted input sequence.
For each input $x_i$  associate a node $v_i$ and add an extra node $v_{n+1}$. Define the weight of the edge from $v_i$ to $v_j$ as the cost of clustering $x_i,\dots,x_{j-1}$ in one cluster, which is $CC(i,j-1)$.
Each edge weight is computed in constant time (by Lemma~\ref{lemma:precompute}) and  the edge weights satisfy the monge concave property by construction.
Finally, to compute the optimal clustering use Schieber's algorithm to compute the lowest weight path with $k$ edges from $v_1$ to $v_{n+1}$.
\begin{Theorem}
  Computing an optimal $k$-Means clustering of an input of size $n$ for given $k = \Omega(\lg n)$ takes $n 2^{O(\sqrt{\lg \lg n \lg k})}$ time using $O(n)$ space.
  \label{thm:fast_clustering}
\end{Theorem}

It is relevant to briefly consider parts of Schieber's algorithm and how it relates to $k$-Means clustering, in particular a \textit{regularized} version of the problem.
Schieber's algorithm relies crucially on algorithms that given a directed acyclic graph where the weights satisfy the concave property computes a minimum weight path in $O(n)$ time \cite{Wilber,Klawe}.
Note the only difference in this problem compared to above, is that the search is not restricted to paths of $k$ edges only.
As noted in \cite{Aggarwal1994}, if the weights are integers then the algorithm solves the Monge Concave Directed Graph Problem in $n \lg U$ time where $U$ is the largest absolute value of a weight.
In the reduction from 1D $k$-Means the weights are not integers and we must take care.


\subsection{Regularized 1D $k$-Means and an $O(n \lg U)$ time algorithm for $k$-Means}
\label{sec:reg}
Consider a regularized version of the $k$-Means clustering problem where instead of providing the number of clusters $k$ we additionally specify a penalty per cluster and ask to minimize the cost of the clustering plus the penalty $\lambda$ for each cluster used.
Formally, the problem is as follows: Given  $\X=\{x_1,...,x_n\}\subset \R$ and $\lambda$, compute the optimal regularized clustering:
\begin{equation*}
  \argmin_{k, \M=\{\mu_1,...,\mu_k\}} \sum_{x\in\X}\min_{\mu\in\M}(x-\mu)^2 + \lambda k
\end{equation*}

If we set $\lambda=0$ the optimal clustering has cost zero and use a cluster for each input point.
If we let $\lambda$ increase towards infinity, the optimal number of clusters used in the optimal solution monotonically decreases towards one (zero clusters is not well defined).
Let $d_\tmin$ be the smallest distance between input points and for simplicity assume all the input points are distinct. The optimal cost of using $n-1$ clusters is then  $d_\tmin^2/2$.
When $\lambda > \lambda_{n-1} = d_\tmin^2/2 $ it is less costly to use only $n-1$ clusters compared to using $n$ clusters since the benefit of using one more cluster is smaller than the cost of a cluster.
Letting $\lambda$ increase again inevitably leads to a miminum value $\lambda_{n-2} > \lambda_{n-1}$ such that for $\lambda > \lambda_{n-2}$ using only $n-2$ clusters gives a better optimal cost than using $n-1$ clusters.
Following the same pattern $\lambda_{n-2}$ is the difference between the optimal cost using $n-2$ clusters and $n-1$ clusters.

Continuing this way yields the very interesting sequence $0 < \lambda_{n-1} \leq \dots \lambda_2 \leq \lambda_1$ that encodes the only relevant choices for the regularization parameter $\lambda$,
where $\lambda_i =^{\mathrm{def}} \OPT_{i} - \OPT_{i+1}$ and $\OPT_i$ is the cost of an optimal $k$-Means clustering with $i$ clusters.
Note that the $O(nk)$ algorithm actually yields $\lambda_1, \dots, \lambda_{k-1}$ since it computes the optimal cost for all $k'\leq k$.
It is an interesting open problem if one can compute this sequence in $n \lg^{O(1)} n$ time, since this encodes all the relevant information about the input instance using linear space, and from that the 1D $k$-Means clustering can be reported in $O(n)$ time for any $k$.

In the reduction to the directed graph problem, adding a cost of $\lambda$ for each cluster used corresponds to adding $\lambda$ to the weight of each edge.
Note that the edge weights still satisfy the concave property.
Thus, solving the regularized version of $k$-Means clustering correponds to finding the shortest path (of any length) in a directed acyclic graph where the weights satisfy the concave property.
By the algorithms in \cite{Wilber,Klawe} this takes $O(n)$ time.
\begin{Theorem}
  Computing an optimal regularized 1D $k$-Means clustering of a sorted input of size $n$ takes $O(n)$ time.
  \label{thm:fast_reg_clustering}
\end{Theorem}
Now if we actually use $\lambda_{k}$ as the cost per cluster in regularized 1D $k$-Means, or any $\lambda \in [\lambda_{k}, \lambda_{k+1}]$ there is an optimal regularized cost solution that uses $k$ clusters which is also an optimal $k$-Means clustering.

This leads to an algorithm for 1D $k$-Means based on a (parametric) binary search for such a parameter $\lambda$ starting with an initial interval of $[\lambda_1, \lambda_{n-1}]$ (both easily computed).
However, there are a few things that need to be considered to ensure that such a binary search can return an optimal clustering.
Most importantly, it may be the case that $\lambda_k = \lambda_{k+1}$
This occurs when the decrease in the clustering cost for going from $k$ to $k+1$ clusters is the same as the decrease when going from $k-1$ to $k$ clusters.
In this case the regularized cost with $\lambda=\lambda_{k} = \lambda_{k+1}$  using $k+1$ and $k$ clusters are the same.
This generalizes to ranges $\lambda_{j+1} < \lambda_j= \lambda_{j-1} = \dots = \lambda_i < \lambda_{i-1}$ for some $j > i+1$.
In this case the optimal regularized cost of using any $k \in \{i+1,\dots,j\}$ is found at $\lambda_{j}$ and even if we can find this value we still need a way to extract an optimal regularized clustering with exactly $k$ clusters.
Fortunately, it is straight forward to compute both an optimal cost regularized clustering with the fewest clusters and an optimal regularized clustering with the most clusters and an optimal regularized clustering with any number of clusters in between these two \cite{Aggarwal1994, Schieber}.

In conclusion a binary search on $\lambda$ ends when the algorithm probes a value where the the number of clusters possible in the optimal regularized clustering contains $k$.
In the worst case this interval is a point, so with $\lg U$ bits integer coordinate input points, this takes in $O(\lg U)$ steps in the worst case.

  \begin{Theorem}
    Computing an optimal $k$-Means clustering of an input of size $n$ takes $O(n \lg U)$ time and $O(n)$ space.
\end{Theorem}


\section{Space Reduction for Dynamic Programming and 1D $k$-Means Reporting Data Structure}
\label{subsec:lowspace}
The space consumption of the $O(kn)$ time dynamic programming algorithm is $O(kn)$ for storing the tables $D$ and $T$ (Equation \ref{equ:rec_cost}, \ref{equ:rec_argmin}). This can easily be reduced to $O(n)$ by not storing $T$ at all and only storing the last two considered rows $D$ during the algorithm (computation of each row of $D$ depends only on the values of the previous row). However, doing this has the downside that we can no longer report an optimal $k$-Means clustering, only the cost of one. This problem can be handled by  a space reduction technique of Hirschberg \cite{Hirschberg} that reduce the space usage of the $O(kn)$ dynamic programming algorithm to just $O(n)$ while maintaining $O(kn)$ running time. For completeness we have given this constuction in the appendix.
This is the standard way of saving space for dynamic programming algorithms and is also applied for a monge concave problem in \cite{GolinZ10}. The problem is that it does cost a constant factor in the running time.

However, for the 1D $k$-Means problems (and other monge concave problems),  we actually do not need to bog down the dynamic programming algorithm at all.
So instead  simply do as follows:
First run the dynamic programming algorithm for $k+1$ clusters, storing only the last two rows as well as the last column of the dynamic programming table $D$. Remeber that the $i$'th entry in the last column stores the cost of the optimal $k$-Means clustering with $i$ clusters i.e. $\OPT_i$. This takes $O(kn)$ time uses linear space and returns $\OPT_1,\dots,\OPT_k,\OPT_{k+1}$.
Given these optimal costs, to extract an optimal clustering compute $\lambda_k = \OPT_{k}-\OPT_{k+1}$ and apply the regularized $k$-Means clustering algorithm with $\lambda_k$.
By our analysis in Section \ref{sec:reg}, this $\lambda_k$ yields an optimal regularized clustering of size $k$, and the algorithm directly returns an actual optimal clustering.
Of course this works for any $k'\leq k$, storing only the optimal $k$-Means costs $\OPT_1,\dots,\OPT_k,\OPT_{k+1}$ we can report an optimal clustering for any $k' \leq k$ in $O(n)$ time, completely forgoing the need to store dynamic programming tables.
We directly get
\begin{Theorem}
  Computing an optimal $k$-Means clustering of a sorted input of size $n$ takes $O(k n)$ time and  $O(n)$ space.
\end{Theorem}
To compute the cost of the optimal clustering for all $k' \leq k$ we keep the last column of the cost matrix $D$ which requires an additional $O(k) = O(n)$ space.

\begin{Theorem}
  There is a data structure that uses $O(n)$ space that can report the optimal 1D $k$-Means clustering for any $k$ in $O(n)$ time.
\end{Theorem}
Note that the best preprocessing time we know for this data structure is $O(n^2)$ by simply running the dynamic programming algorithm with $k=n$, but of course it can be constructed using a smaller $k$ at the cost of only being able to report optimal clusterings up to size $k$.
Furthermore,  this construction does not depend on actual using the $k$-Means cost function, just the Monge Concave property and the directed graph shortest path problem and thus generalize to all problems considered by Schieber \cite{Schieber}.

\section{Extending to More Distance Measures}
\label{sec:extensions}
In the following we show how all the above algorithms generalize to all
 Bregman Divergences and the sum of absolute distances while retaining
 the same running time and space consumption.

\subsection{Bregman Divergence Clusterings}
First, let us remind ourselves what a  Bregman Divergence and a Bregman Clustering is.
Let $f$ be a differentiable real-valued strictly convex function.
The Bregman Divergence $D_f$ defined by $f$ is defined as
\begin{align*}
  D_f(x,y) = f(x) - f(y) - \nabla_f(y)(x-y)
\end{align*}

\paragraph{Bregman Clustering.}
The Bregman Clustering problem as defined in \cite{Banerjee:jmlr}, is to find a set of centers,  $\M=\{\mu_1,...,\mu_k\}$, that minimizes
\begin{align*}
  \sum_{x \in \mathcal{X}} \min_{\mu\in \M} D_f(x,\mu)
\end{align*}
Notice that the cluster center is the second argument of the Bregman
Divergence. This is important since Bregman Divergences are not in general symmetric.

For the purpose of 1D clustering, we mention two important properties of Bregman Divergences.
For any Bregman Divergence, the unique element that minimizes the summed distance to a multiset of elements is the mean of the elements, exactly as it was for squared Euclidian distance.
This is in a sense the defining property of Bregman Divergences \cite{Banerjee:jmlr}.
The second property is the linear separator property, which is crucial for clustering with Bregman Divergences but also for Bregman Voronoi Diagrams \cite{Banerjee:jmlr,Boissonnat2010}.
\paragraph{Linear Separators For Bregman Divergences.}
For all Bregman divergences, the locus of points that are equidistant to two fixed points $\mu_1, \mu_2$ in terms of a Bregman divergence is given by $\{x \in \mathcal{X} \mid D_f(x,p) = D_f(x,q)\}$.
Plugging in the definition of a Bregman Divergence this is
\begin{align*}
  \{x \in \mathcal{X} \mid D_f(x,p) = D_f(x,q)\}
  = \{x \in \mathcal{X} \mid x (\nabla_f(\mu_1) -  \nabla_f(\mu_2))
  = f(\mu_1)-\mu_1 \nabla_f(\mu_1) - f(\mu_2) + \mu_2 \nabla_f(\mu_2)\}
\end{align*}
which is a hyperplane.
The points $\mu_1, \mu_2$ sit on either side of the hyperplane and the Voronoi cells defined by Bregman divergences are connected.

This means, in particular, that between any two points in 1D, $\mu_1 < \mu_2$,
there is a hyperplane (point) $h$ with $\mu_1 < h < \mu_2$ and all points smaller than $h$ are closer to $\mu_1$ and all points larger than $h$ are closer to $\mu_2$.
For 1D Bregman Divergence Clustering it means the optimal clusters correponds to intervals (as was the case in 1D $k$-Means).
We capture what we need from this observation in a simple ``distance'' lemma:
\begin{lemma}
  \label{lemma:bregman_distance}
  Given two fixed real numbers $\mu_1 < \mu_2$, then for   any point $x_r\geq \mu_2$, we have $D_f(x_r,\mu_1) > D_f(x_r,  \mu_2)$,
  and for any point $x_l \le \mu_1$ we have $D_f(x_l,\mu_1) <  D_f(x_l, \mu_2)$
\end{lemma}

\paragraph{Computing Optimal Cluster Costs for Bregman Divergences.}
Since the mean minizes Bregman Divergences, the centroids used in optimal clusterings are unchanged compared to the $k$-Means case.
The prefix sums idea used to implement the data structure used for Lemma \ref{lemma:precompute} generalizes to Bregman Divergences as observed in \cite{Nielsen2014OptimalIC} (under the name Summed Area
Tables). The formula for computing the cost of grouping the points $x_i,\dots, x_j$ in one cluster is as follows. Let $\mu_{i,j} = \frac{1}{j-i+1}\sum_{\ell=i}^j x_\ell$ be the arithmetic mean of the points $x_i,\dots,x_j$, then
\begin{align*}
  CC(i,j) & = \sum_{\ell=i}^j D_f(x_\ell, \mu_{i,j})\\
  & = \sum_{\ell=i}^j f(x_\ell) - f(\mu_{i,j}) - \nabla_f (\mu_{i,j})(x_\ell - \mu_{i,j})\\
  & = \left(\sum_{\ell=i}^j f(x_\ell)\right) - (j-i+1) f(\mu_{i,j})  -  \nabla_f (\mu_{i,j}) \left(\left(\sum_{\ell=i}^j x_\ell\right) - (j-i+1) \mu_{i,j}\right)
\end{align*}
%
%
Rearranging terms this is
$$
\left( \sum_{\ell=i}^j f(x_\ell)\right) - (j-i+1) f(\mu_{i,j})\\
-  \nabla_f (\mu_{i,j}) \left(\left(\sum_{\ell=i}^j x_\ell\right) - (j-i+1) \mu_{i,j}\right).
$$
%
Thus the Bregman Divergence cost of a consecutive subset of input points can be computed in in constant time with stored prefix sums for $x_1, \dots ,x_n$ and $f(x_1), \dots ,f(x_n)$.

\subsubsection{Totally Monotone Matrix and Monge Concave property for Bregman Divergences}

\begin{lemma}
  For any Bregman Divergence $D_f$ the induced cluster cost $CC(i, j)$ function is concave.
  Formally, for any $a < b$ and $u < v$, it holds that:
  $$
  CC(v, b) + CC(u, a)  \leq  CC(u,b) + CC(v, a).
  $$
\end{lemma}
\begin{proof}
  We start by handling the special case where $v > a$. In this case, we have by definition that $CC(v,a)=0$, thus we need to show that $CC(v, b) + CC(u, a)  \leq  CC(u,b)$.
  This is the case since any point amongst $x_u,\dots,x_b$ is included in at most one of $x_v,\dots,x_b$ and $x_u,\dots,x_a$ (since $a < v$).
  Thus $CC(v, b) + CC(u, a)$ is the cost of taking two disjoint and consecutive subsets of the points $x_u,\dots,x_b$ and clustering the two sets using the optimal choice of centroid in each.
  Clearly this cost is at most the cost of clustering all the points using one centroid (both groups could use the same cluster center as the cluster center for all the points).

  We now turn to the general case where $u < v \leq a < b$. Let $\mu_{v,a}$ be the mean of $x_{v},\dots, x_{a}$ and $\mu_{u,b}$ be the mean of $x_u,\dots,x_b$ and assume that $\mu_{v,a} \leq \mu_{u,b}$ (the other case is symmetric).
  Finally, let $CC(w,z)_\mu = \sum_{\ell=w}^z D_f(x_\ell,\mu)$ denote the cost of grouping the elements $x_w,\dots,x_z$ into a cluster with centroid $\mu$.

  Split the cost $CC(u,b)$ into the cost of the elements $x_u,\dots,x_{v-1}$ and the cost of the elements $x_v,\dots,x_b$ i.e.
  \begin{align*}
    CC(u,b) &
    = \sum_{\ell=u}^{v-1} D_f(x_\ell, \mu_{u,b}) + \sum_{\ell=v}^b D_f(x_\ell, \mu_{u,b}) = CC(u,v-1)_{\mu_{u,b}} + CC(v,b)_{\mu_{u,b}}.
  \end{align*}
  We trivially get  $CC(v,b)_{\mu_{u,b}} \geq CC(v,b)$ since $CC(v,b)$ is the cost using the optimal centroid.

  Since $\mu_{v,a} \leq \mu_{u,b}$ and all elements $x_u,\dots,x_{v-1}$ are less than or equal to $\mu_{v,a}$ (since $\mu_{v,a}$ is the mean of points $x_{v},\dots,x_a$ that all are greater than $x_u,\dots,x_{v-1}$), then by Lemma \ref{lemma:bregman_distance},
  $$
  CC(u,v-1)_{\mu_{u,b}} \geq CC(u, v-1)_{\mu_{v,a}}
  $$
  Adding the cluster cost $CC(v, a)$ to both sides of this inequality we get that
    \begin{align*}
      CC(u,v-1)_{\mu_{u,b}} + CC(v, a)
      \geq CC(u,v-1)_{\mu_{v,a}} + CC(v, a)
      =  CC(u, a)_{\mu_{v,a}} \geq CC(u, a)
    \end{align*}
    %
    Combining the results, 
    \begin{align*}
       CC(v, b) + CC(u, a)
      \leq CC(v,b)_{\mu_{u,b}} + CC(u,v-1)_{\mu_{u,b}} + CC(v, a)
      = CC(u,b) + CC(v,a).
    \end{align*}
\end{proof}

It follows that all the results achieved for 1D $k$-Means presented earlier generalize to any Bregman Divergence.

\subsection{$k$-Median Clustering}
\label{sec:kmediod}
For the $k$-Medians problem we replace the the sum of squared
Euclidian distances with the sum of absolute distances.  Formally, the
$k$-Medians problem is to compute a clustering, $\M=\{\mu_1,...,\mu_k\}$,  minimizing
\begin{align*}
	\sum_{x\in\X}\min_{\mu\in\M}\lvert x-\mu \rvert
\end{align*}
Note that in 1D, all $L_p$ norms are the same and reduce to this case. Also note that the minimizing centroid for a cluster is no longer the mean
of the points in that cluster, but the median. To solve this problem, we change the centroid to be the median, and if the number
of points is even, we fix the median to be the exact middle point between the two middle elements, making the choice of centroid unique.

As for Bregman Divergences, we need to show that we can compute the
cluster cost $CC(i,j)$ with any Bregman Divergence in constant time. Also, we need to
compute the centroid in constant time and  argue that the
cost is monge moncave which implies the implicit matrix $\COST_i$ is totally monotone.
The arguments are essentially the same, but for completeness we briefly cover them below.

\paragraph{Computing Cluster Costs for Absolute Distances.}
Not surprisingly, prefix sums still allow constant time
computation of $CC(i,j)$. 
Let $m_{i,j} = \frac{j+i}{2}$, and compute the centroid as $\mu_{i,j} = \frac{x_{\lfloor m_{i,j} \rfloor} + x_{\lceil m_{i,j} \rceil} }{2}$ then
%
$$
CC(i,j)  = \sum_{\ell=i}^j \lvert x_\ell - \mu_{i,j} \rvert 
= \sum_{\ell=i}^{\lfloor{m_{i,j}\rfloor}}  \mu_{i,j} -x_\ell   + \sum_{\ell=1+\lfloor{m_{i,j}\rfloor}}^{j}  x_\ell - \mu_{i,j}
$$
which can be computed in constant time with access to a prefix sum
table of $x_1,\dots,x_n$. This was also observed in \cite{Nielsen2014OptimalIC}.

\paragraph{Monge Concave - Totally Monotone Matrix.}
The monge concave and totally monotone matrix argument above for
Bregman Divergences (and for squared Euclidian distance) remain valid
since first of all, we still have $x_u,\dots,x_{v-1} \leq \mu_{v,a}$
as $\mu_{v,a}$ is the median of points all greater than
$x_u,\dots,x_{v-1}$.  Furthermore, it still holds that when $\mu_{v,a}
\leq \mu_{u,b}$ and all elements $x_u,\dots,x_{v-1}$ are less than or
equal to $\mu_{v,a}$, then $CC(u,v-1)_{\mu_{u,b}} + CC(v, a) \geq
CC(u,v-1)_{\mu_{v,a}} + CC(v, a) = CC(u, a)_{\mu_{v,a}}$.
%
%
It follows that the  algorithms we specified for 1D $k$-Means generalize to the 1D $k$-Median problem.


%
%

\section{Experiments}
To asses the  practicality of 1D $k$-Means algorithms we have implemented and compared different versions of the  $O(kn)$ time algorithm and the $O(n\lg U)$ time binary search algorithm.
We believe Schieber's $n 2^{O(\sqrt{ \lg \lg n \lg k})}$ time algorithm \cite{Schieber} is mainly of theoretical interest, at least the constants involved in the algorithm seem rather large.
Similarly for the  $O(k \sqrt{n \log n})$ time algorithm of \cite{Aggarwal1994}. Both are slightly similar to the $O(n \lg U)$ time algorithm in the sense they work by trying to find the right $\lambda$ value for the regularized $k$-Means problem so that it uses $k$ clusters. However, to get data independent bounds (independent of $U$) the algorithms become much more complicated. 
For reasonable values of $k$ we expect the simple $O(kn)$ time algorithm to be competitive with these more theoretical efficient algorithms.
Our main interest is seeing how the  binary search algorithm compares.

\subsection{Data Sets}
For the experiments we consider two data sets.
\begin{description}
  \item[Uniform:] The \texttt{Uniform} data set is, as the name suggest, created by sampling the required number of points uniformly at random between zero and one.
  \item[Gaussian Mixture:] The \texttt{Gaussian} Mixture data set is created by defining 16 Gaussian Distributions, each with variance 100, and means placed one million apart.
    Each data point required is created by uniformly at random sampling one of the 16 gaussians and then sample a point from that.
\end{description}

\subsection{Algorithm Setup}
\label{sec:experiments}

\paragraph{Dynamic Programming.}
For the dynamic programming algorithm we have implemented the $O(kn\lg n)$ and the $O(kn)$ time algorithm.
The difference lies only in how the monotone matrix search is performed (Section \ref{sec:fully_monotone}).
The latter using the SMAWK algorithm and the former a simple divide and conquer approach.
This is done because the SMAWK algorithm may in fact be slower in practice.
Both versions have been implemented using only $O(n)$ space in two different ways, the first version only storing the last row of the dynamic programming matrix and the second version reducing space using the Hirschberg space saving technique (Section \ref{subsec:lowspace}).


\paragraph{Binary Search.}
For the  binary search algorithm we have considered two algorithms for regularized $k$-Means, namely the algorithms in \cite{Wilber} and \cite{Klawe}.
The invariants and pseudo-code in \cite{Klawe} are not entirely correct and does not directly turn into working code.
We were able to fix these issues and made the algorithm run, however, our implementation of \cite{Klawe} is clearly slower than our implementation of \cite{Wilber}.
Therefore, for the experiments we only consider the linear time algorithm from \cite{Wilber}.

Remember that $\lambda_i$ is the minimum value for $\lambda \ge 0$ where the optimal regularized cost of using $i$ clusters is less than the optimal regularized cost of using $i+1$ clusters.
The standard implementation is a binary search on the regularization coefficient ($\lambda$), that ends when a $\lambda$ that gives an optimal $k$-Means clustering is found. 
The starting range for $\lambda$ is $[0,M]$ where $M$ is the cost of clustering all points in one cluster ($\textrm{OPT}_1$).
For $\lambda = 0$, the optimal clustering uses $n$ clusters and has a cost of zero.
For $\lambda = M$, the optimal clustering uses one cluster, as its regularized cost is $2M$, and using at least $2$ clusters costs strictly more than $2M$ (assuming any cluster costs more than $0$).
The algorithm simply maintains the current range for $\lambda$, $\lambda_\textrm{low}$ and  $\lambda_\textrm{high}$. We denote by $k_\textrm{low}$, and $k_\textrm{high}$ the number of clusters used in the optimal regularized cost solution at $\lambda_\textrm{low}$ and  $\lambda_\textrm{high}$ respectively.
Note that $k_\textrm{high} < k_\textrm{low}$, and let $I_k = \{k_\textrm{high}, \dots, k_\textrm{low}\}$ be the current interval for the number clusters that includes $k$.
Finally, let $c_\textrm{low}$ and $c_\textrm{high}$ denote the unregularized  cost of using $k_\textrm{low}$ and $k_\textrm{high}$ clusters respectively. 

Many different values for $\lambda$ give the same regularized clustering (same number of clusters), which means if we attempt two or more such values for $\lambda$, we make almost no progress.
Also, the middle of the current range for $\lambda$ does not in any way translate to the middle of the range $k_\textrm{high}, k_\textrm{low}$.
Early experiments showed that the standard binary search actually had this issue. %

To overcome this issue, we modified the search to work in a more data dependent way to guarantee progress by picking the next $\lambda$  as the value  where the regularized costs of using $k_\textrm{high}$ clusters and $k_\textrm{low}$ clusters are the same.
At this position we are guaranteed that there is an optimal cost regularized clustering using $k'$ clusters where $k'$ is strictly contained in $k_\textrm{high}$ and $k_\textrm{low}$ which we prove below.


To be able to do this we update the algorithm to also maintain $k_\textrm{high}, k_\textrm{low}$, and $c_\textrm{high}, c_\textrm{low}$.  Given these computing $\lambda'$ is straight forward, namely.
$$
\lambda' = (c_\textrm{high} - c_\textrm{low})/(k_\textrm{low} - k_\textrm{high})
$$
which is the average increase in the optimal cost per cluster added.
\begin{lemma}
  \label{lem:intersection}
  Let $\lambda' = (c_\textrm{high}-c_\textrm{low})/(k_\textrm{low} - k_\textrm{high})$, at $\lambda'$ there is an optimal cost regularized clustering with $k'$ clusters  where $k'$ is contained in $\{k_\textrm{high}+1,\dots, k_\textrm{low}-1\}$
\end{lemma}
\begin{proof}
Consider the regularized $k$-Means cost as a function of $\lambda$. This is a piecewise linear function since it corresponds to the the minimum of $n$ lines, the $k$'th line representing the regularized cost of an optimal clustering using $k$ clusters, having slope $k$ and intercept (value at $\lambda=0$) the optimal (non-regularized) cost of using $k$ clusters.
For each line the subset of the real line where this line is the minimun of the $n$ lines is a line segment, and the line segment's end points correspond to when the number of clusters in the optimal cost regularized clustering change.

This means that the value $\lambda'$ we try in the binary search is the intersection of the two line segments corresponding to $\lambda_\textrm{low}$ and $\lambda_\textrm{high}$.
At this $\lambda'$ the regularized cost of using $k_\textrm{low}$ clusters and $k_\textrm{high}$ clusters is the same.
For all $k'$ betwen $k_\textrm{high}$ and $k_\textrm{low}$ the regularized cost of the optimal regularized clustering at $\lambda'$ is at most that of using $k_\textrm{low}$ or $k_\textrm{high}$ clusters (that are the same).

Consider a number of clusters $k'$ strictly between $k_\textrm{low}$ and $k_\textrm{high}$ and the associated line for using $k'$ clusters. This line must intersect the line for $k_\textrm{low}$ at some value $x \leq \lambda'$, otherwise the line can never be the minimum line since for $x < \lambda'$ it would be above the line for $k_\textrm{low}$ because of its lower slope, and for $x > \lambda'$ it would be above the line for $k_\textrm{high}$ due to higher slope.
Finally, since the line for $k'$ clusters has lower slope than the line for $k_\textrm{low}$ and intersects it at $x \leq \lambda'$ the regularized cost at $\lambda'$ of using $k'$ clusters is at most the cost of using $k_\textrm{low}$.
\end{proof}
Notice, that if $k_\textrm{high}$ clusters achieve the smallest regularized cost at $\lambda'$, then the regularized cost of all $k'\in I_k$ at $\lambda'$ are the same and the optimal $k$-Means clustering is found.

The unfortunate side effect of using this strategy is the running time of $O(n\lg U$) may no longer hold. That is easily remedied by for instance every $t=O(1)$ steps query the midpoint of the remaining range, however no such trick is used in our implementation. 

We refer to the standard binary search algorithm as \emph{wilber-binary} and the intersection binary search as \emph{wilber-interpolation}.
For a comparison of these two algorithms,  see Figure \ref{fig:runtime_binsearch}.
In these plots the search based on the intersection/interpolation is clearly superior on both data sets and and we consider only this version of the binary search  in the following runtime comparison with the dynamic programming algorithms.

\setcounter{figure}{0}
\begin{figure}[t]
\centering
\begin{subfigure}{.5\textwidth}
  \centering
  \includegraphics[width=.9\linewidth]{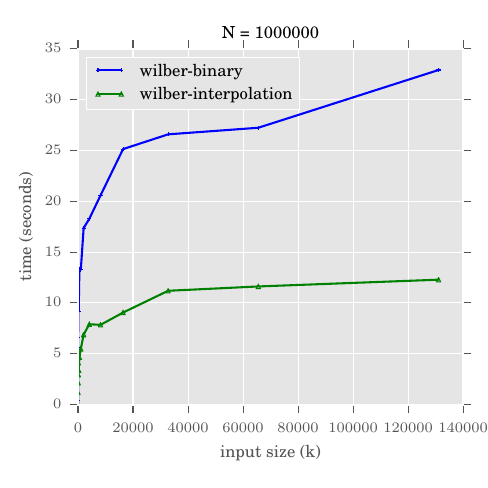}
  \caption{Running times on \texttt{Gaussian}  Data Set}
  \label{fig:runtime_n}
\end{subfigure}%
\begin{subfigure}{.5\textwidth}
  \centering
  \includegraphics[width=.9\linewidth]{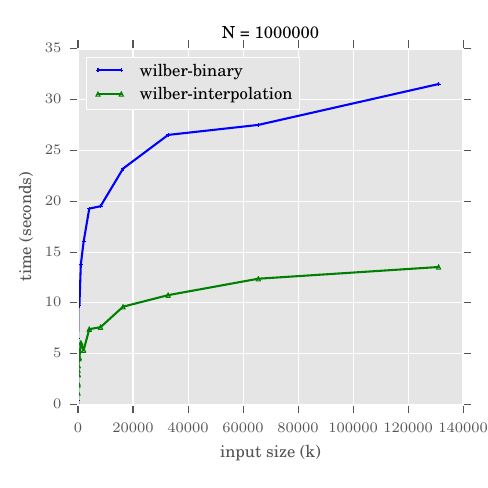}
  \caption{Running times on \texttt{Uniform} Data Set}
  \label{fig:runtime_k}
\end{subfigure}
\caption{Running time comparison binary search algorithms as a function of $k$}
\label{fig:runtime_binsearch}
\end{figure}


%

\subsection{Algorithm Comparison}
We refer to the $O(kn\log n)$ algorithm with space saving mentioned in Section \ref{sec:fully_monotone} as \texttt{DP-monotone} and the variant using the Hirschberg technique as \texttt{DP-monotone-hirsch}.
Similarly we call the $O(kn)$ algorithm (also with space saving) \texttt{DP-linear}, and \texttt{DP-linear-hirsch} when using the Hirschberg technique.
Finally we denote the $O(n \log U)$ algorithm \texttt{Wilber} since the regularized $k$-Means is an implementaion of Wilber's algorithm \cite{Wilber}. It should be noted  that the space saving technique really is necessary as $n$ and $k$ grow, since otherwise the space grows with the product, which is quite undesirable.

The experiments do \textbf{not} include the time for actually reporting a clustering which gives the \texttt{DP-linear} an advantage over the other algorithm since it would require an invocation of the \texttt{Wilber} algorithm to report a clustering, while  \texttt{DP-linear-hirsch} and \texttt{Wilber}  extracts an optimal clustering to report during the algorithm and would have no extra cost.

%
Figures \ref{fig:runtime_n_uniform}, \ref{fig:runtime_k_uniform}, and \ref{fig:runtime_k_gaussian} show the running time as a function of $n$ or $k$ on the \texttt{Uniform} data set.
The performance of the dynamic programming algorithms are as expected. These algorithm always fill out a table of size $kn$ and are thus never better than the worst case running time.
The plots also reveal that for the values of $n$ tested, the $O(kn \lg n)$ is in fact superior to the $O(kn)$ time algorithm albeit the difference is not large.
As the plots show, when $k$ grows, the \texttt{Wilber} algorithm is much faster than the other algorithms (even when $k=20$ for the smallest $n$ we tried). This is both true for \texttt{Uniform} and the \texttt{Gaussian} data set.
It is also worth noting that even for moderate values of $n$ and $k$ the space quickly goes in the order of gigabytes for the dynamic programming solutions, if we maintain the entire table.

%
Notice that the dynamic programming algorithm can report the clustering cost for all $k' \leq k$ using an additive $O(k)$ space by always keeping the final column of the dynamic programming table.
On the other hand \texttt{Wilber} cannot report the costs of all clusterings, but it can report \emph{some} of them, as it searches for a cluster cost $\lambda$ that yields a clustering with $k$ clusters.
For each $\lambda$ that is attempted, the cost of an optimal clustering using a different $k$ may be reported.
For practitioners that want to see the plot of the cost of the best clustering as a function of either $k$ or $\lambda$, the \texttt{Wilber} algorithm might still be sufficient, as it does provide points on that curve and one can even make an interactive plot: if a desired point on the curve is missing just compute it and add more points to the plot. For a new $\lambda$ this takes linear time and for a new $k$ we need to binary search the interval between the currently stored nearest neighbours.

The simple conclusion is that for these kinds of data, the binary search algorithm is superior even for moderate $n$ and $k$, and for large $n$, $k$ it is the only choice. If one prefers the guarantee of the dynamic programming algorithm, implementing \texttt{Wilber} allows for saving a non-trivial constant factor in the running time for linear space algorithms and allows reporting an optimal clustering for any $k'\leq k$ in linear time.

\begin{figure}[t]
  \centering
  \begin{subfigure}[b]{.3\textwidth}
    \centering
    \includegraphics[width=\textwidth]{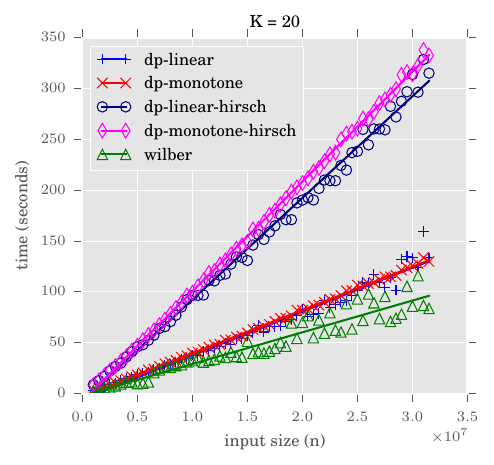}
    \caption{Running times as a function of $n$ on the \texttt{Uniform} data set}
    \label{fig:runtime_n_uniform}
  \end{subfigure}\hfill%
  \begin{subfigure}[b]{.3\textwidth}
    \centering
    \includegraphics[width=\textwidth]{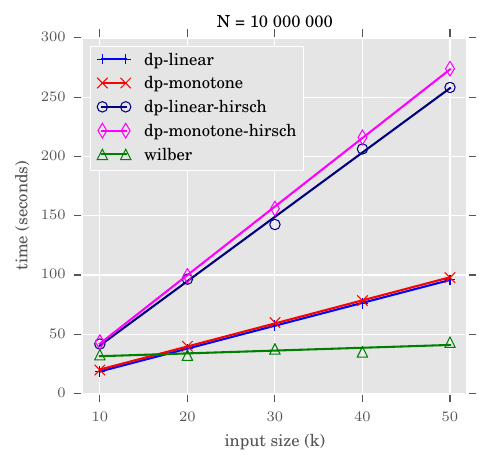}
    \caption{Running times as a function of $k$ on the \texttt{Uniform} data set}
    \label{fig:runtime_k_uniform}
  \end{subfigure}\hfill%
  \begin{subfigure}[b]{.3\textwidth}
    \centering
    \includegraphics[width=\textwidth]{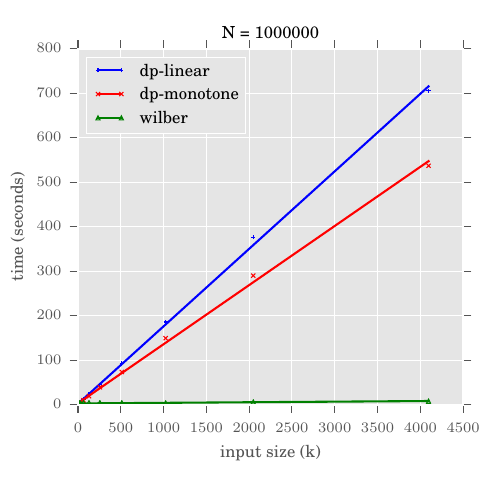}
    \caption{Running times as a function of $k$ on the \texttt{Gaussian} data set}
    \label{fig:runtime_k_gaussian}    
  \end{subfigure}
  \caption{Running time comparison of different 1D $k$-Means algorithms.
    dp-* are different versions of dynamic programming and wilber is the binary search based on interpolation}
\label{fig:runtime_all}
\end{figure}




\section*{Final Remarks}
We have given an overview of 1D $k$-Means algorithms,  generalized them to new measures, shown the practical performance of several algorithm variants including a simple way of boosting the binary search algorithm. We have defined the obvious regularized version of 1D $k$-Means which is important not only for fast algorithms based on binary search but also linear space solutions for reporting of actual optimal clusterings based on the dynamic programming algorithms.
We see a few important problems left open
\begin{itemize}
  \item Is there an $n \lg^{O(1)} n$ time algorithm for 1D $k$-Means or maybe even an $n\lg^{O(1)} k$ time algorithm (if the input is sorted)?
  \item Is there an  $n \lg^{O(1)} n$ or even $n^{2-O(1)}$ time algorithm for computing the optimal $k$-Means costs for all $k=1,\dots,n$ yielding  the sequence $\lambda_1,\dots, \lambda_{n-1}$ that encodes all relevant information for the given 1D $k$-Means instance?
  \item What is the running time of the search algorithm using the tweak we employed? An easy bound is $O(n^2)$ but we were not able to get such lousy running time in practice. In fact it seemed to be a really good heuristic for picking the next query point in the binary search.
  \item The dynamic programming algorithm with a running time of $O(kn)$ can rather easily be parallelized to run in $O((kn \lg n)/p)$ for $p$ processors, by parallelizing the monotone matrix search algorithm (not SMAWK but the simple $O(n\lg n)$
    divide and conquer algorithm).
    For the binary search algorithm, it is possible to try and improve
    the $\lg U$ to $\lg_p U$ for $p$ processors, but it would be much
    better if one could parallelize the linear time 1D regularized k-Means algorithm (or a near linear time version of it).
\end{itemize}

We wish to thank Pawel Gawrychowski for pointing out important earlier work 
on \concave  property.
\bibliographystyle{plain}
\bibliography{main}
\appendix
\section{Reducing space usage of $O(kn)$ time dynamic programming algorithm using Hirschberg}
Remeber that each row of $T$ and $D$ (Equation \ref{equ:rec_cost}, \ref{equ:rec_argmin}) only refers to the previous row. Thus one can clearly ``forget''row $i-1$ when we are done computing row $i$
In the following, we present an algorithm that avoids the table $T$ entirely.

The key observation is the following: Assume $k>1$ and that for every prefix $x_1,\dots,x_m$, we have computed the optimal cost of clustering $x_1,\dots,x_m$ into $\lfloor k/2 \rfloor$ clusters. Note that this is  the set of values stored in the $\lfloor k/2 \rfloor$'th row of $D$. Assume furthermore that we have computed the optimal cost of clustering every suffix $x_m,\dots,x_n$ into $k-\lfloor k/2 \rfloor$ clusters. Let us denote these costs by $\tilde{D}[k-\lfloor k/2 \rfloor][m]$ for $m=1,\dots,n$. Then clearly the optimal cost of clustering $x_1,\dots,x_n$ into $k$ clusters is given by:
\begin{equation}
\label{eq:midsplit}
\begin{split}
D[k][n] =&  \min_{j=1}^n D[\lfloor k/2 \rfloor][j] + \tilde{D}[k-\lfloor k/2 \rfloor][j+1].
\end{split}
\end{equation}
The main idea is to first compute row $\lfloor k/2 \rfloor$ of $D$ and row $k-\lfloor k/2 \rfloor$ of $\tilde{D}$ using linear space. From these two, we can compute the argument $j$ minimizing \eqref{eq:midsplit}. We can then split the reporting of the optimal clustering into two recursive calls, one reporting the optimal clustering of points $x_1,\dots,x_j$ into $\lfloor k/2 \rfloor$ clusters, and one call reporting the optimal clustering of $x_{j+1},\dots,x_n$ into $k-\lfloor k/2 \rfloor$ clusters. When the recursion bottoms out with $k=1$, we can clearly report the optimal clustering using linear space and time as this is just the full set of points.

From Section~\ref{sec:fully_monotone} we already know how to compute row $\lfloor k/2 \rfloor$ of $D$ using linear space: Simply call SMAWK to compute row $i$ of $D$  for $i=1,\dots, \lfloor k/2 \rfloor$, where we throw away row $i-1$ of $D$ (and don't even store $T$) when we are done computing row $i$. Now observe that table $\tilde{D}$ can be computed by taking the points $x_1,\dots,x_n$ and reversing their order by negating the values. This way we obtain a new ordered sequence of points $\tilde{X} = \tilde{x}_1 \leq \tilde{x}_2 \leq \cdots \leq \tilde{x}_n$ where $\tilde{x}_i = -x_{n-i+1}$.
Running SMAWK repeatedly for $i=1,\dots,k-\lfloor k/2 \rfloor$ on the point set $\tilde{X}$ produces a table $\hat{D}$ such that $\hat{D}[i][m]$ is the optimal cost of clustering $\tilde{x}_1,\dots,\tilde{x}_m = -x_n,\dots,-x_{n-m+1}$ into $i$ clusters. Since this cost is the same as clustering $x_{n-m+1},\dots,x_n$ into $i$ clusters, we get that the $(k-\lfloor k/2\rfloor)$'th row of $\hat{D}$ is identical to the $i$'th row of $\tilde{D}$ if we reverse the order of the entries.

To summarize the space saving algorithm for reporting the optimal clustering, does as follows: Let $L$ be an initially empty output list of clusters. If $k=1$, append to $L$ a cluster containing all points. Otherwise ($k > 1$), use SMAWK on $x_1,\dots,x_n$ and $-x_n,\dots,-x_1$ to compute row $\lfloor k/2 \rfloor$ of $D$ and row $k-\lfloor k/2 \rfloor$ of $\tilde{D}$ using linear space (by evicting row $i-1$ from memory when we have finished computing row $i$) and $O(kn)$ time. Compute the argument $j$ minimizing \eqref{eq:midsplit} in $O(n)$ time. Evict row $\lfloor k/2 \rfloor$ of $D$ and row $k-\lfloor k/2 \rfloor$ of $\tilde{D}$ from memory. Recursively report the optimal clustering of points $x_1,\dots,x_j$ into $\lfloor k/2 \rfloor$ clusters (which appends the output to $L$). When this terminates, recursively report the optimal clustering of points $x_{j+1},\dots,x_n$ into $k-\lfloor k/2 \rfloor$ clusters. When the algorithm terminates, $L$ contains the optimal clustering of $x_1,\dots,x_n$ into $k$ clusters.

At any given time, the algorithm uses $O(n)$ space: We evict all memory used to compute the value $j$ minimizing \eqref{eq:midsplit} before recursing. Furthermore, we complete the first recursive call (and evict all memory used) before starting the second recursive call. Finally, for the recursion, we do not need to copy $x_1,\dots,x_j$. It suffices to remember that we are only working on the subset of inputs $x_1,\dots,x_j$.

Let $F(n,k)$ denote the time used by the above algorithm to compute an optimal clustering of $n$ sorted points into $k$ clusters. Then there is a constant $C>0$ such that $F(n,k)$ satisfies the recurrence:
$
F(n,1) \leq C n,
$
and for $k > 1$:
$$
F(n,k) \leq \max_{j=1}^n F(j,\lfloor k/2 \rfloor) + F(n-j, k-\lfloor k/2 \rfloor) + C n k .
$$
We claim that $F(n,k)$ satisfies $F(n,k) \leq 3Ckn$. We prove the claim by induction in $k$. The base case $k=1$ follows trivially by inspection of the formula for $F(n,1)$. For the inductive step $k>1$, we use the induction hypothesis to conclude that $F(n,k)$ is bounded by
%
%
\begin{align*}
&\max_{j=1}^n 3Cj\lfloor k/2 \rfloor  + 3C(n-j)(k-\lfloor k/2\rfloor)  + C n k  \\
&\leq \max_{j=1}^n 3Cj\lceil k/2 \rceil  + 3C(n-j)\lceil k/2 \rceil  + C n k  \\
& = 3Cn \lceil k/2 \rceil  + Ckn .
\end{align*}
For $k>1$, we have that $\lceil k/2 \rceil \leq (2/3)k$, therefore:
\begin{align*}
F(n,k) &\leq 3Cn (2/3)k + Ckn = 3Ckn.
\end{align*}
%

\end{document}